\documentclass[11pt]{article}
\usepackage{graphicx, amsmath, amsthm, amssymb, tikz}
\usepackage{thmtools}
\usepackage{float}
\usetikzlibrary{matrix}
\usetikzlibrary{positioning}
\usetikzlibrary{shapes}
\usetikzlibrary{decorations.pathreplacing}
\usetikzlibrary{calc,patterns,angles,quotes}
\usetikzlibrary{intersections}

\usepackage{makecell}
\usepackage{enumitem} 
\usepackage{adjustbox}
\usepackage{array}
\newcolumntype{P}[1]{>{\centering\arraybackslash}p{#1}}

\usepackage{fullpage}
\usepackage[hypertexnames=false]{hyperref}
\hypersetup{
    colorlinks=true,
    allcolors=blue
}
\usepackage{algorithm}
\usepackage{algpseudocode}
\usepackage{nicematrix}
\usepackage{blkarray}
 
\newcommand{\ket}[1]{|#1\rangle}

\newcommand{\cA}{\mathcal{A}}
\newcommand{\cB}{\mathcal{B}}
\newcommand{\cC}{\mathcal{C}}
\newcommand{\cD}{\mathcal{D}}

\newcommand{\cG}{\mathcal{G}}
\newcommand{\cH}{\mathcal{H}}

\newcommand{\cM}{\mathcal{M}}

\newcommand{\cR}{\mathcal{R}}
\newcommand{\cS}{\mathcal{S}}
\newcommand{\cT}{\mathcal{T}}

\newcommand{\cX}{\mathcal{X}}

\newcommand{\sC}{\mathsf{C}}
\newcommand{\sD}{\mathsf{D}}

\newcommand{\sQ}{\mathsf{Q}}
\newcommand{\sR}{\mathsf{R}}

\newcommand{\set}[1]{\left\{ #1 \right\}}
\newcommand{\rbra}[1]{\left(#1\right)}
\newcommand{\cbra}[1]{\left\{#1\right\}}
\newcommand{\zone}{\set{0,1}}

\newtheorem{theorem}{Theorem}[section]

\newtheorem{remark}[theorem]{Remark}
\newtheorem{lemma}[theorem]{Lemma}

\allowdisplaybreaks

\title{Complexity of learning matchings and half graphs via edge queries}
\author{Nikhil S.~Mande\thanks{University of Liverpool, UK {\tt nikhil.mande@liverpool.ac.uk}}
\and 
Swagato Sanyal\thanks{University of Sheffield, UK {\tt swagato.sanyal@sheffield.ac.uk}
}
\and
Viktor Zamaraev\thanks{University of Liverpool, UK {\tt viktor.zamaraev@liverpool.ac.uk}
}
}

\date{}
\begin{document}

\maketitle

\begin{abstract}
    The problem of learning or reconstructing an unknown graph from a known family via partial-information queries arises as a mathematical model in various contexts.
    The most basic type of access to the graph is via \emph{edge queries}, where an algorithm may query the presence/absence of an edge between a pair of vertices of its choosing, at unit cost.

    While more powerful query models have been extensively studied in graph reconstruction, comparatively less is known about exact reconstruction from a promised graph family using only edge queries. In this paper we study the edge query complexity of learning a hidden bipartite graph, or equivalently its bipartite adjacency matrix, in the classical as well as quantum settings.
    We focus on learning matchings and half graphs, which are graphs whose bipartite adjacency matrices are a row/column permutation of the identity matrix and the lower triangular matrix with all entries on and below the principal diagonal being 1, respectively.
    \begin{itemize}
        \item For matchings of size $n$, we show a tight deterministic bound of $n(n-1)/2$ and an asymptotically tight randomized bound of $\Theta(n^2)$. A quantum bound of $\Theta(n^{1.5})$ was shown in a recent work of van Apeldoorn et al.~[ICALP'21].

        \item For half graphs whose bipartite adjacency matrix is a column-permutation of the $n \times n$ lower triangular matrix,
        we give tight $\Theta(n \log n)$ bounds in both deterministic and randomized settings, and an $\Omega(n)$ quantum lower bound. 
        We also observe that this learning problem is equivalent to sorting with threshold comparisons. 

        \item For general half graphs, 
        we observe that the problem is equivalent to a natural generalization of the famous nuts-and-bolts problem, leading to a tight $\Theta(n \log n)$ randomized bound.
        We also present a simple quicksort-style method that instantiates to a $O(n \log^2 n)$ randomized algorithm and a tight $O(n \log n)$ quantum algorithm.
    \end{itemize}
\end{abstract}

\newpage

\section{Introduction}\label{sec:intro}
A graph learning (also known as a graph reconstruction) problem is a formalization of the task of reconstructing a network from its partial observations. 
This task may arise when observing the entire network is costly or not possible, but knowledge about the network topology is desirable. For example, it appears in evolutionary biology, genetics, bioinformatics, telecommunication networks, and the theory of chemical reaction networks (see references below). A natural objective is to minimize the number of partial observations needed to reconstruct the graph.

Partial observations are usually modeled as queries that ask for some local or global information about parts of the graph. Different types of queries are motivated by different contexts. Perhaps the most extensively studied queries are:

\begin{itemize}
    \item \emph{Edge detection queries}. For a given set of vertices of the graph, this query outputs 0 if there are no edges between any pair of vertices in the set and it outputs 1 otherwise. Such queries arise in the context of the theory of chemical reaction networks, where one has a set of chemicals some pairs of which may or may not react.
    The goal is to learn which pairs of chemicals react with each other by doing a small number of experiments each consisting of mixing a set of chemicals and observing the existence of a reaction.
    Many studies investigated this type of queries for reconstructing graphs from various graph classes;
    see, for example,~\cite{GK97, ABKRS04, BGK05, AC08, AB19} and  references therein.

    \item \emph{Edge counting queries}. For a given set of vertices, this query outputs the number of edges in the subgraph induced by the given set of vertices. It can be used to model genome sequencing via multiplex polymerase chain reaction (multiplex PCR). 
    See, for example,~\cite{GK00, BGK05, CK10} and references therein.

    \item \emph{Distance queries}. For a given pair of vertices, this query returns the length of a shortest path between the two vertices. This type of queries was introduced in~\cite{BEEHHMR06} to model the network topology discovery in the context of telecommunication or peer-to-peer networks when the entire network topology is not available, but node-to-node messages can be used to estimate distances. See, for example,~\cite{MZ13,KMZ15,RLYW21,MZ23,BG23,KZ24} and references therein.
\end{itemize}

\noindent
Further types of queries for reconstructing graphs from a given class that were studied in the literature include 
Parity queries (returns the parity of the number of edges in the subgraph induced by the input set of vertices) \cite{MS22} and
Maximal Independent Set queries (returns a maximal independent set in the subgraph induced by the input set of vertices) \cite{KOT24}.

Another possible type of query one may consider, and the one we study in this work, is \emph{edge queries}, which, given a pair of vertices, outputs whether or not they are adjacent. This is arguably the most basic model to consider. It is easy to see that each of the above-mentioned queries can implement an edge query by querying the underlying pair of vertices.
Quantum edge-query algorithms and lower bounds have also been studied for various graph decision and optimization problems in the adjacency-matrix model~\cite{BDFLS04,DHHM06}.
Therefore an algorithm in any of the above query models can learn a graph using the \emph{naive strategy} that makes $\binom{n}{2}$ \emph{edge queries} to learn adjacencies between all pairs of vertices.
All of these types of queries are powerful, witnessed by algorithms for reconstructing graphs from various graph classes, as shown in the studies cited above.
This raises the question of whether the generality of the more powerful queries is needed to outperform the naive strategy.
To examine this question we study the problem of learning graphs via edge queries in deterministic, randomized, and quantum settings.
We show that in some cases edge queries can be used to learn graphs much more efficiently than the naive strategy.

\subsection{Our results, techniques, and comparison with prior work}

We focus on edge query complexity of learning graphs from two specific classes of bipartite graphs: matchings and half graphs. We assume that the bipartition of a graph is given as input and we want to learn edges between the vertices in the two parts of the graph. Alternatively, we will see this problem as learning an $n \times n$ matrix from a given family of matrices that correspond to bipartite adjacency matrices of the graphs from a class.

To state our results we introduce some notation.
We will say that a bipartite graph $G = ([n],[n], E)$ is a \emph{matching} graph if the degree of every vertex in $G$ is exactly 1. Let $H=([n],[n], E)$ be a bipartite graph such that $(i,j) \in [n] \times [n]$ is an edge in $H$ if and only if $i \leq j$; a bipartite graph $G = ([n],[n], E)$ is a \emph{half graph} if it is isomorphic to $H$.
We denote by $I_n$ the $n \times n$ identity matrix and by $L_n$ the $n \times n$ lower triangular 0-1 matrix whose entries below or on the main diagonal are 1s and all other entries are 0s.
We denote by $\cM_n$ the family of matrices that are obtained from $I_n$ via row or column permutations; by $\cC_n$ we denote the family of matrices that are column permutations $L_n$, and by $\cH_n$ we denote the family of matrices that are obtained from $L_n$ via row or column permutations.
Interpreting matrices in these classes as bipartite adjacency matrices of bipartite graphs with bipartition consisting of two $n$-vertex sets, we will often think about $\cM_n$, $\cC_n$, and $\cH_n$ as the corresponding families of matching graphs, half graphs whose bipartite adjacency matrices are column permutations of $L_n$ (we refer to such graphs in $\cC_n$ as \emph{column-permuted} half graphs), and half graphs with arbitrary labelings of vertices in both parts, respectively.

For a family $\cX_n$ of bipartite graphs (equivalently, the family of their bipartite adjacency matrices), we also denote by $\sD(\cX_n)$, $\sR(\cX_n)$, $\sQ(\cX_n)$ the number of edge queries that any deterministic, randomized, and quantum algorithm, respectively, needs to make in the worst case to learn a graph from $\cX_n$. In the randomized and quantum settings, an algorithm is required to output the correct answer with probability at least $2/3$.
With this notation at hand, we are now ready to state and discuss our results.

\subsubsection{Matchings}
Below is our main result regarding the query complexities of learning hidden matchings.

\begin{restatable}{theorem}{mainResultMatching}
\label{thm-matching}
    For every natural $n$, we have
    $$
        \sD(\cM_n) = \frac{n(n - 1)}{2}, \qquad 
        \sR(\cM_n) = \Theta(n^2), \qquad
        \sQ(\cM_n) = \Theta(n^{1.5}).
    $$
\end{restatable}

The quantum query complexity bounds follow from a relatively recent result of van Apeldoorn et al.~\cite{AGLNWW21}.
We also note that the precise deterministic query complexity in Theorem \ref{thm-matching} is not very surprising as such a precise complexity of $n(n-1)$ is known for learning a matching in a very similar setting where the bipartition is not given \cite{Aig88}. 
We include a proof of our deterministic bound for completeness.

The randomized lower bound is new. A direct application of Yao's minimax principle followed only by counting leaves gives merely $\Omega(\log |\cM_n|)=\Omega(n\log n)$. Nevertheless, the uniform distribution on $\cM_n$ is hard enough for the desired $\Omega(n^2)$ bound once one also counts pairs of matchings related by a column swap. We present the lower bound using Aaronson's classical variant of the adversary method~\cite{Aar04}; after the proof, we spell out the equivalent direct Yao argument. The same adversary relation, suitably restricted, proves the stronger statement that recovering only the parity of every matched neighbor still requires $\Omega(n^2)$ randomized queries.

A related work to ours is that of Berzina et al.~\cite{BDFLS04}, where an $\Omega(n^{1.5})$ quantum query lower bound was shown for deciding whether a graph contains a matching. We focus solely on the reconstruction problem in this paper.

\subsubsection{Column-permuted half graphs}

Below is our main result regarding learning column-permuted half graphs.
\begin{restatable}{theorem}{mainResultColPermutedHalfGraphs}
\label{thm-col-permuted-half-graphs}
    For every natural $n$, we have
    $$
    \sD(\cC_n) = \Theta(n \log n), \qquad \sR(\cC_n) = \Theta(n \log n).
    $$
    Moreover, $\sQ(\cC_n)=\Omega(n)$ and $\sQ(\cC_n)=O(n\log n)$.
\end{restatable}

The upper bounds in Theorem \ref{thm-col-permuted-half-graphs} follow easily by applying binary search to each column to find the topmost 1-entry of the column.

The deterministic and randomized lower bounds follow from a standard counting argument, which we include as Lemma \ref{lem:large-range-query-lower-bound} for completeness.
The quantum $\Omega(n)$ lower bound is obtained using the adversary method~\cite{Amb02}. While an $\Omega(n \log n)$ comparison-based quantum lower bound for sorting holds true~\cite{HNS02}, we were unable to adapt that proof to our setting.
We conjecture that $\sQ(\cC_n)=\Theta(n\log n)$. The classical output-counting argument does not directly extend to quantum queries.

We show in Section \ref{sec:threshold-sorting}
that the query complexity of learning column-permuted half graphs is equivalent to the complexity of sorting an unknown list $X$, obtained by applying an arbitrary permutation to $[n]$, with \emph{threshold} comparisons of the form ``Is $X[j] \geq i$?'' for any $i, j \in [n]$.

\subsubsection{Half graphs}

We show in Section \ref{sec:perfectly-interleaved} that the query complexity of learning half graphs is equivalent to the complexity of the perfectly-interleaved bipartite sorting problem that was recently studied by Goswami and Jacob~\cite{GJarxiv, GJ24, GJapprox}. This problem is a natural generalization of the classical nuts-and-bolts problem~\cite{Raw92}, where there is no promised matching between the
nuts and the bolts.

Goswami and Jacob~\cite{GJarxiv, GJapprox} obtained a tight randomized $O(n \log n)$ algorithm for perfectly interleaved bipartite sorting. Together with our equivalence, this yields $O(n\log n)$ upper bounds in both the randomized and quantum settings; the lower bounds below show that both upper bounds are tight.

A randomized tight lower bound of $\Omega(n \log n)$ immediately follows from the same lower bound for learning column-permuted half graphs. 
The proof of the comparison-based $\Omega(n \log n)$ lower bound for sorting~\cite[Theorem~2]{HNS02} can be immediately seen to yield a quantum $\Omega(n \log n)$ lower bound on the edge query complexity for learning half graphs.

Beyond the equivalence above and the resulting transfer of the known tight randomized bound, our additional contribution in this setting is a simple quicksort-type randomized algorithm with complexity $O(n \log^2 n)$.
Interestingly, while the naive quicksort-type algorithm for the classic nuts-and-bolts problem pivots on a nut along with its matching bolt, our algorithm sorts the nuts separately, and the bolts separately, both using a simple quicksort-style algorithm. While our algorithm is a log-factor slower than that of Goswami and Jacob that alternates between sorting nuts and bolts in its steps, our algorithm has some advantages: first, our algorithm and analysis are much simpler and essentially follow a textbook-style quicksort analysis; second, by replacing a certain subroutine in our algorithm with a quantum subroutine, we obtain a tight $O(n \log n)$ quantum algorithm. 

\begin{restatable}{theorem}{mainResultGeneralHalfGraphs}
\label{thm-general-half-graphs}
    There exist simple quicksort-type randomized and quantum algorithms for learning graphs in $\cH_n$ with randomized query complexity of $O(n \log^2 n)$, and a tight quantum query complexity of $O(n \log n)$, respectively.
\end{restatable}

Koml{\'{o}}s, Ma, and Szemer{\'{e}}di~\cite[Section~4]{KMS98} claim that their $O(n \log n)$-cost deterministic algorithm for the classic nuts-and-bolts problem also works in our setting, but we were unable to verify this.

We conclude this section with a summary of best-known edge query complexity bounds presented in Table~\ref{table: bounds} (excluding the above-mentioned bound), and the following intriguing observation. As one can observe, in the (deterministic, randomized, or quantum) edge query model, learning matchings is much harder than learning half graphs.
Interestingly, this is in stark contrast with the fact that in the two-party public-coin randomized communication model testing adjacency in matching graphs is much easier (constant cost) than in half graphs (non-constant cost) \cite{HWZ22}.

\begin{table}[H]
    \centering
    \setlength{\tabcolsep}{4pt}

    \begin{tabular}{| P{.29\textwidth} | P{.22\textwidth} | P{.22\textwidth} | P{.18\textwidth} |}
         \hline
        Description of graph class & Deterministic bounds & Randomized bounds & Quantum bounds \\ [0.5ex] 
         \hline\hline
        matchings & $\frac{n(n-1)}{2}$ & $\Theta(n^2)$ & $\Theta(n^{1.5})$\\
         \hline
        \makecell{column-permuted\\ half graphs} & $\Theta(n \log n)$ & $\Theta(n \log n)$ & $\Omega(n)$, $O(n \log n)$\\
        \hline
        half graphs & $\Omega(n \log n), O(n^2)$ & $\Theta(n \log n)$ & $\Theta(n \log n)$\\
         \hline
    \end{tabular}

    \caption{\label{table: bounds}
    Edge query complexity of matchings and half graphs.}
\end{table}

\subsection{Organization}
We introduce the necessary notation and preliminaries in Section~\ref{sec:prelims}. Section \ref{sec:matching} is devoted to edge query complexity of learning matchings. In Section \ref{sec:learning-half-graphs} we study edge query complexity of learning half graphs and column-permuted half graphs; we also establish the equivalence of these problems with sorting problems. Finally, in Section~\ref{sec:conc}, we conclude with a brief discussion of the scope of our techniques, as well as some open questions and future directions.

\section{Preliminaries}\label{sec:prelims}
All logarithms in this paper are base 2. For a positive integer $n$, we use the notation $[n]$ to denote the set $\cbra{1, 2, \dots, n}$. For a string $x \in \zone^n$ and $i \in [n]$, we use the notation $x_i$ to denote the $i$'th bit of $x$. We use the notation $|x|$ to denote the Hamming weight of $x$, which is $|\cbra{i \in [n] : x_i = 1}|$.
 
\subsection{Query complexity}\label{subsec:prelims-query-complexity}

Throughout this subsection we assume $n$ is a positive integer, $\cD \subseteq \zone^n$ is a finite set, $\cR$ is an arbitrary finite set, and $f : \cD \to \cR$ is a function.

A \emph{decision tree}, also called a \emph{query algorithm}, is a binary tree whose leaf nodes are labeled by elements of $\cR$, each internal node is labeled by an index $i \in [n]$ and has two outgoing edges, labeled $0$ and $1$.
On an input $x \in \zone^n$, the tree's computation proceeds from the root down to a leaf as follows: query $x_i$ as indicated by the node's label and follow the edge indicated by the value of $x_i$.
Continue this way until reaching a leaf, at which point the value of the leaf is output.

A query algorithm is said to compute $f$ if the output of the tree on input $x$ equals $f(x)$ for all $x \in \cD$.
The cost of a query algorithm is the number of queries made on a worst-case input, which is exactly the depth of the corresponding tree.
Formally, the query complexity of $f$, denoted $\sD(f)$, is defined as
\[
\sD(f) := \min_{T : T~\text{is a decision tree computing}~f} \textnormal{depth}(T).
\]
A randomized decision tree is a distribution over deterministic decision trees. We say a randomized decision tree computes $f$ with error $1/3$ if for all $x \in \cD$, the probability of it outputting $f(x)$ is at least $2/3$. The depth of a randomized decision tree is the maximum depth of a deterministic decision tree in its support.
Define the randomized query complexity of $f$ as follows.
\[
\sR(f) := \min_{\substack{T : T~\textnormal{is a randomized decision tree}\\\textnormal{that computes}~f~\textnormal{up to error}~1/3}}
\textnormal{depth}(T).
\]

We refer the reader to~\cite{BW02, NC16} for the basics of quantum computing and quantum query complexity. A quantum query algorithm $\mathcal{A}$ for $f$ begins in a fixed initial state $\ket{\psi_0}$ in a finite-dimensional Hilbert space, applies a sequence of unitaries $U_0, O_x, U_1, O_x, \dots, U_T$, and performs a measurement. Here, the initial state $\ket{\psi_0}$ and the unitaries $U_0, U_1, \dots, U_T$ are independent of the input. The unitary $O_x$ represents the ``query'' operation, and does the following for each basis state:
it maps $\ket{i}\ket{b}$ to $\ket{i}\ket{b + x_i \mod 2}$ for all $i \in [n]$ and all $b \in \zone$.
The algorithm then performs a measurement whose outcomes are labeled by elements of $\cR$ and outputs the observed label.
We say that $\mathcal{A}$ is a bounded-error quantum query algorithm computing $f$ if for all $x \in \cD$ the probability of outputting $f(x)$ is at least $2/3$. The quantum query complexity of $f$, denoted by $\sQ(f)$, is the least number of queries required for a quantum query algorithm to compute $f$ with error $1/3$. In the expected-case setting, an algorithm is allowed to perform intermediate measurements (and terminate based on the observed values), required to always output the correct answer, and the cost is the expected number of applications of $O_x$, for the worst-case input $x$. 

We require the following lower bound on deterministic and randomized query complexities of a function whose range is large.
The proof is standard and follows along the lines of the standard comparison-based lower bound for sorting, but we include a proof in the appendix for completeness.

\begin{restatable}{lemma}{countingLowerBound}
\label{lem:large-range-query-lower-bound}
    Let $n$ be a positive integer, $\cD \subseteq \zone^n$ and $\cR$ be finite sets. Let $f : \cD \to \cR$ be a surjective function.
    Then
    \[
    \sD(f) = \Omega\rbra{\log|\cR|}, \qquad \sR(f) = \Omega\rbra{\log|\cR|}.
    \]
\end{restatable}

We remark that an analogous statement, where the domain is non-Boolean and each query has a small set of outcomes, admits a similar proof.

We require the following modification of Grover's search algorithm~\cite{Grover96} due to Boyer et al.~\cite{BBHT98}.

\begin{lemma}\label{lem:bbht}
    Let $n$ be a positive integer and let $0^n \neq x \in \zone^n$. There exists a quantum query algorithm that uses $O(\sqrt{n/|x|})$ queries in expectation and outputs an index $i \in [n]$ such that $x_i = 1$.
\end{lemma}
The analogous classical statement, stated below, is easy to see: simply sample indices uniformly at random until a 1 is seen.
\begin{lemma}\label{lem:sample}
    Let $n$ be a positive integer and let $0^n \neq x \in \zone^n$. There exists a randomized query algorithm that uses $O(n/|x|)$ queries in expectation and outputs an index $i \in [n]$ such that $x_i = 1$.
\end{lemma}

\subsection{Learning hidden graphs}
In this paper, we study the query complexity of learning (i.e., reconstructing) a hidden bipartite graph that comes from a known class. We assume that the two parts of the bipartite graph are of the same size and the (ordered) bipartition is given; without loss of generality, we assume that each part is the set $[n]$. For a known class $\cG_n$ of such graphs, an algorithm receives an unknown (i.e., hidden) graph $G=([n],[n],E) \in \cG_n$ and its goal is to recover $E$ by making as few \emph{edge queries} as possible. An edge query is an ordered pair of vertices $(i,j) \in [n] \times [n]$, to which an oracle answers 1 if the vertex $i \in[n]$ in the left part of the graph is adjacent to the vertex $j \in[n]$ in the right part of the graph; otherwise the oracle answers 0.

Throughout this paper, we identify each graph in $\cG_n$ with its \emph{bipartite} adjacency matrix, i.e., the $n \times n$ 0-1 matrix whose $(i,j)$'th entry is 1 if and only if $(i,j)$ is an edge in the graph. Thus, the graph learning problem can be seen as the problem of learning a hidden $n \times n$ 0-1 matrix from a known class by accessing as few of its entries as possible.

From this point of view, the problem of learning a hidden graph can be seen within the framework described in Section~\ref{subsec:prelims-query-complexity} as the problem of computing the identity function $\mathsf{Id}_{\cG_n} : \cG_n \to \cG_n$, where the input is viewed as a Boolean vector in $\zone^{n^2}$ representing the bipartite adjacency matrix of the input graph.
Depending on the setting (deterministic, randomized, or quantum) we measure the complexity of the graph class $\cG_n$ as the corresponding query complexity of computing $\mathsf{Id}_{\cG_n}$.
Formally, we define the deterministic, randomized, and quantum query complexity of identifying a hidden graph from $\cG_n$, denoted $\sD(\cG_n)$, $\sR(\cG_n)$, $\sQ(\cG_n)$, respectively, as $\sC(\cG_n) := \sC(\mathsf{Id}_{\cG_n})$, where $\sC \in \cbra{\sD, \sR, \sQ}$.

In the literature regarding randomized and quantum algorithms for graph reconstruction (see papers cited in Section~\ref{sec:intro}), one sometimes considers a cost measure that is the expected number of queries, maximized over all inputs, where an algorithm is required to always output the correct answer.
Worst-case bounded-error complexity is at most a constant times this expectation-based cost: run the algorithm for three times its worst-case expected cost, abort and output arbitrarily if it has not terminated, and apply Markov's inequality to argue about the large success probability.
We remark that all of our randomized and quantum upper bounds are in the ``weaker'' expected-case model, while our lower bounds are all in the ``stronger'' worst-case setting.

\section{Learning matchings}\label{sec:matching}
In this section we are interested in the query complexity of learning a hidden matching in a graph with a known bipartition. We consider the deterministic, randomized and quantum settings.
Recall that $\cM_n$ denotes the class of all $n \times n$ permutation matrices, i.e., bipartite adjacency matrices of matchings. Our main result of this section is the following.

\mainResultMatching*

We devote the rest of this section to proving all the bounds in the theorem above.

\subsection{Deterministic bounds}
\begin{lemma}\label{lem:matchingub-det}
    There exists a deterministic query algorithm of cost $\frac{n(n-1)}{2}$ that learns a hidden matching from $\cM_n$.
\end{lemma}

\begin{proof}
    Consider the naive greedy algorithm that iteratively finds the neighbor of each vertex on the left side.
    In order to do this for a specific vertex $v$, it checks for adjacency with all but one unmatched neighbors on the right side. Either this process encounters the neighbor of $v$, or the last vertex on the right side has to be its neighbor. Thus, when $i$ vertices have been matched, we can discover the next edge in a maximum of $n - i - 1$ queries. Summing this over $i$ from $0$ to $n-1$ gives the upper bound of $\frac{n(n-1)}{2}$.
\end{proof}

\begin{lemma}\label{lem:matchinglb-det}
    Any deterministic query algorithm that learns a hidden matching from $\cM_n$ must make at least $\frac{n(n-1)}{2}$ queries.
\end{lemma}

\begin{proof}
    Towards a contradiction, consider an algorithm $\cA$ of query cost less than $\frac{n(n-1)}{2}$ that learns a hidden matching from $\cM_n$. We may assume that the algorithm makes no redundant queries, that is, every internal node in the corresponding decision tree has two children, and every leaf contains at least one consistent matching. Consider the output at the leaf of the decision tree that corresponds to all edge queries being answered as 0. Since there is at least one consistent matching at this leaf, we may assume that it is the Identity matrix (by permuting the row of the matrix or, equivalently, by relabeling the vertices in the right side of the graph). By our assumption, none of the diagonal entries are queried on this path because all the query outcomes are 0. If we can show existence of a $2\times2$ submatrix with $1$'s on the main diagonal whose off-diagonal entries have not been queried, this would imply the existence of another matching reaching this leaf by flipping the 0's and 1's in this submatrix. This would yield a contradiction since the leaf has only one output.

    Whenever $\cA$ queries a pair on this path of the form $(i, j)$ with $i > j$, we assume that it has also queried the pair $(j, i)$ at no extra cost. There are $\frac{n(n-1)}{2}$ entries of the matrix above the main diagonal. If there is one such entry that is not queried, say $(i, j)$ with $i < j$, then by the above assumption the entire submatrix comprising the entries $(i, i), (i, j), (j, i), (j, j)$ is unqueried, yielding the required contradiction. Thus this path must make at least $\frac{n(n-1)}{2}$ edge queries.
\end{proof}

\subsection{Randomized bounds}
Recall that the randomized query complexity of learning a hidden graph in $\cM_n$ is the randomized query complexity of the identity function $f = \mathsf{Id}_{\cM_n} : \cM_n \to \cM_n$.
The intuition behind our $\Omega(n^2)$ randomized lower bound is as follows: 
Let $\cR$ denote a relation consisting of pairs of matrices from $\cM_n$, in which each matrix is one column swap away from its partner, i.e.~two columns are interchanged. Let $\cR_{i, j}$ denote the subrelation of $\cR$ in which each pair differs in the entry $(i, j)$.
Intuitively, a randomized algorithm for computing $f$ must distinguish a constant fraction of the pairs in $\cR$. Moreover, a single query to the entry $(i,j)$ can distinguish at most $|\cR_{i, j}|$ pairs in $\cR$. Thus, one would expect $|\cR|/\max_{i, j}|\cR_{i,j}|$ to be a lower bound on the randomized query complexity of $f$.
Aaronson's variant of Ambainis' adversary method adapted to give randomized lower bounds~\cite[Theorem~5]{Aar04} captures this intuition and this is what we use. The version we state below is from~\cite[Section~3]{AKPVZ21}.

\begin{theorem}[{\cite[Theorem~5]{Aar04}}]\label{thm: classical adversary bound}
    Let $S \subseteq \zone^{N}$, let $H$ be a finite set, and let $f : S \to H$ be a function. Let $R : S \times S \to \mathbb{R}_{\geq 0}$ be a real-valued function such that $R(x, y) = R(y, x)$ for all $x, y \in S$ and $R(x, y) = 0$ whenever $f(x) = f(y)$. For $x \in S$ and $i \in [N]$, define
    $$
    \theta(x, i) = \frac{\sum_{y \in S} R(x, y)}{\sum_{y \in S : x_i \neq y_i} R(x, y)},
    $$
    where $\theta(x, i)$ is undefined if the denominator is 0. Define
    $$
    \mathsf{CRA}(f) = \max_R\min_{\substack{x, y \in S, i \in [N] :\\R(x, y) > 0, x_i \neq y_i}} \max\cbra{\theta(x, i), \theta(y, i)}.
    $$
    Then, $\sR(f) = \Omega(\mathsf{CRA}(f))$.
\end{theorem}

We use the above theorem with $f = \mathsf{Id}_{\cM_n}, S = H = \cM_n$, which is exactly our setting of learning a hidden matching. With function $R$ we simulate a relation consisting of pairs of matrices that are hard to distinguish: a pair of matchings $(M, M')$ is in the relation if and only if $M'$ can be obtained from $M$ by doing a single column swap. Intuitively, these are the pairs of inputs that are hardest to distinguish by edge queries because they differ only in the four entries of a $2\times2$ submatrix. A formal statement and its proof are below.

\begin{theorem}\label{thm:rand-lb-matching}
The randomized query complexity of learning a hidden matching from $\cM_n$ is $\Theta(n^2)$.
\end{theorem}
\begin{proof}
The upper bound is trivial. 
We use Theorem~\ref{thm: classical adversary bound} to show the lower bound. In Theorem~\ref{thm: classical adversary bound}, set $N = n^2$ (we view an element in $\zone^{N}$ as a $[n] \times [n]$ 0-1 matrix), let $S$ be the set of $n \times n$ permutation matrices (bipartite adjacency matrices of matchings), and let $f : S \to S$ be the identity function, that is, the goal is to recover the input. This is exactly our setting, with query access to the bits of the input. For input matrices $M_1, M_2$, define 
$$
R(M_1, M_2) = \begin{cases}
    1 & \text{if } M_1 \textnormal{ and } M_2 \textnormal{ are one column swap away from each other; }\\
    0 & \textnormal{otherwise}.
\end{cases}
$$
We now analyze the quantity $\mathsf{CRA}(f)$ from Theorem~\ref{thm: classical adversary bound} with this function $R$. Fix arbitrary inputs $M_1, M_2$ with $R(M_1, M_2) > 0$ (by our definition, this means $M_1$ and $M_2$ are one column swap away from each other). Fix $(i, j) \in [n] \times [n]$ such that $M_1(i, j) \neq M_2(i, j)$. Without loss of generality assume that $M_1(i, j) = 0$. We have
$$
\theta(M_1, (i, j)) = \frac{\sum_{M \in S} R(M_1, M)}{\sum_{M \in S : M_1(i, j) \neq M(i, j)} R(M_1, M)} = \binom{n}{2}.
$$
To see the above, first note that the numerator is $\binom{n}{2}$ since any pair of columns in $M_1$ can be swapped. For the denominator, observe that there is only one other column that can be swapped with the $j$'th one in $M_1$ to change the entry $(i,j)$ from 0 to 1: this is the unique column whose $i$'th row has a 1-entry. Thus the denominator is 1.
Hence, $\mathsf{CRA}(f) \geq \binom{n}{2}$, and therefore Theorem~\ref{thm: classical adversary bound} implies $\sR(f) = \Omega(n^2)$.
\end{proof}

Rather than appealing to the classical adversary bound (Theorem~\ref{thm: classical adversary bound}), one could also prove Theorem~\ref{thm:rand-lb-matching} by a direct application of Yao's minimax principle~\cite{Yao77} using a double-counting argument. We sketch a high-level idea of this proof below.

\begin{remark}[Equivalent proof via Yao's principle]
The uniform distribution on $\cM_n$ is already a hard distribution. Let $\Gamma$ be the graph with vertex set $\cM_n$ in which two permutation matrices are adjacent if one can be obtained from the other by a single column swap. Then $\Gamma$ is $d$-regular, where $d=\binom{n}{2}$. Fix a deterministic decision tree $\cA$ of depth $q$ that is correct on a set $S\subseteq \cM_n$ of at least $2n!/3$ inputs, and let $F=\cM_n\setminus S$. Since
\[
2|E(\Gamma[S])| = d|S|-|E_\Gamma(S,F)|
\ge d(|S|-|F|)\ge \frac{dn!}{3},
\]
the subgraph induced by $S$ contains at least $dn!/6$ edges.
For each edge $\{M,M'\}$ of $\Gamma[S]$, the computation paths of $\cA$ on $M$ and $M'$ must diverge; charge this edge to the first query at which they diverge, oriented toward the endpoint containing a 0 in the queried position. At a node that queries the entry $(i,j)$, every matrix reaching that node with a 0 in position $(i,j)$ has exactly one neighbor in $\Gamma$ with a 1 in that position. Thus at most as many edges are charged to the node as there are inputs reaching it. Summing over all nodes, the number of charged edges is at most
\[
\sum_{M\in\cM_n}\operatorname{depth}_{\cA}(M)\le qn!.
\]
Consequently $qn!\ge dn!/6$, so $q\ge d/6=\Omega(n^2)$. Yao's minimax principle~\cite{Yao77} now gives the randomized lower bound.
\end{remark}

We note below that our arguments above can be modified to show the stronger statement that even learning only the parities of the row positions of the $1$'s in each column of a hidden matching has randomized query complexity $\Omega(n^2)$.

\begin{remark}\label{rem:randomized-parity-matching}
Let $M\in\cM_n$ encode the permutation $\sigma_M\in S_n$, where $M(i,j)=1$ if and only if $i=\sigma_M(j)$, and define $z(M)_j=\sigma_M(j)\bmod 2$. The preceding adversary argument in the proof of Theorem~\ref{thm:rand-lb-matching} also shows that recovering $z(M)$ with bounded error requires $\Omega(n^2)$ randomized queries. Indeed, restrict the relation $\cR$ to swaps of two columns whose 1-entries lie in rows of opposite parity. Every matrix then has $\lfloor n/2\rfloor\lceil n/2\rceil=\Theta(n^2)$ related matrices, while, for every entry on which a related pair differs, the endpoint containing a 0 has exactly one related swap that changes that entry. Theorem~\ref{thm: classical adversary bound} therefore gives the claimed lower bound.
\end{remark}

\subsection{Quantum bounds}

An asymptotically tight quantum query complexity bound of $\Theta(n^{1.5})$ for our problem follows from a recent result of van Apeldoorn et al.~\cite{AGLNWW21}. Their Lemma~16 proves the stronger statement that recovering the parity descriptor $z(M)$ from Remark~\ref{rem:randomized-parity-matching} requires $\Omega(n^{1.5})$ quantum queries, which immediately implies the lemma below.
\begin{lemma}[{\cite[Lemma~16]{AGLNWW21}}]\label{lem: quantum matching}
    Let $n$ be a positive multiple of 4. Then the quantum query complexity of learning a hidden matching from $\cM_n$ is $\Omega(n^{1.5})$.
\end{lemma}
The restriction that $n$ be a multiple of 4 does not affect the asymptotic lower bound.
As observed in \cite{AGLNWW21}, this bound is tight due to a matching upper bound achieved with a simple algorithm that, for every vertex on the left side of the bipartition, finds its unique neighbor on the right side using Grover's search algorithm with $O(\sqrt{n})$ queries per vertex. Consequently, the quantum query complexity of learning a hidden matching from $\cM_n$ is $\Theta(n^{1.5})$.

\section{Learning half graphs}\label{sec:learning-half-graphs}
In this section, we consider the problem of learning a hidden half graph or, equivalently, learning a hidden matrix which is a permutation of the lower triangular matrix.
We consider two different settings.
The first setting is where only columns are permuted (equivalently, only rows are permuted), and the second setting is where both rows and columns are permuted. In Section \ref{subsec:sorting}, we show that the former is equivalent to sorting with threshold comparisons, and the latter is equivalent to perfectly interleaved bipartite sorting \cite{GJarxiv, GJapprox}, which is a generalization of the classical nuts-and-bolts sorting problem.
 
Our main results about learning half graphs are proved in Sections \ref{subsec:column-permuted-half-graph-learning} and \ref{subsec:row-column-permuted-half-graph-learning}. We state them below after recalling the necessary notation.
We use $L_n$ to denote the $n \times n$ lower triangular matrix, i.e., the matrix whose entries below (and including) the main diagonal are all 1. 
We denote by $\cC_n$ the set of $n \times n$ matrices that can be obtained by applying an arbitrary column permutation to $L_n$; and by $\cH_n$ the set of $n \times n$ matrices that can be obtained by applying an arbitrary column permutation and an arbitrary row permutation to $L_n$.

\mainResultColPermutedHalfGraphs*

\mainResultGeneralHalfGraphs*

\noindent

\subsection{Equivalence with sorting problems}\label{subsec:sorting}
In this section, we observe that the problems of learning column-permuted half graphs (i.e., graphs from $\cC_n$) and half graphs (i.e., graphs from $\cH_n$) are equivalent to certain sorting problems. We start with the former.

\subsubsection{Sorting with threshold queries}\label{sec:threshold-sorting}
Let $X$ be a hidden list of length $n$ obtained by applying an arbitrary permutation to $[n]$. The goal of an algorithm is to identify $X$. The algorithm has access to an oracle that answers queries of the form ``Is $X[j] \geq i$?'' for any $i, j \in [n]$. We refer to such queries as \emph{threshold queries}. The cost of the algorithm is the number of queries made to the oracle in the worst case. 
Let $\cT_n$ denote the family of all lists obtained as a permutation of $[n]$.
Using the query complexity framework from Section \ref{subsec:prelims-query-complexity},
we define the deterministic, randomized, and quantum query complexities of identifying a hidden list from $\cT_n$ with threshold queries. We denote them as $\sD^t(\cT_n)$, $\sR^t(\cT_n)$, and $\sQ^t(\cT_n)$, respectively.

\begin{lemma}\label{lem:sorting-thr-equiv-colpermuted}
    Let $n$ be a positive integer. Then for all $\sC \in \cbra{\sD, \sR, \sQ}$,
    $$
    \sC^t(\cT_n) = \sC(\cC_n).
    $$
\end{lemma}
\begin{proof}
    To prove the lemma, we establish a bijection between the lists in $\cT_n$ and matrices in $\cC_n$, and show that threshold queries to a list in $\cT_n$ are equivalent to edge queries to the corresponding matrix in $\cC_n$.

    The bijection is defined as follows. To a list $X \in \cT_n$ we associate the matrix $M_X \in \cC_n$, where the $j$-th column of $M_X$ is $0^{n - X[j]}1^{X[j]}$. That is, the top $n - X[j]$ entries are 0s, and the remaining entries are 1s. According to this bijection, a matrix $M \in \cC_n$ corresponds to the list $X_M \in \cT_n$, where for every $j \in [n]$, $X_M[j]$ is equal to the number of 1s in the $j$-th column of $M$.
    To show the equivalence between queries, we observe that for any $X \in \cT_n$ and $M \in \cC_n$, we have 
    $X[j] \geq n - i + 1$ if and only if $M_X(i, j) = 1$, and
    $M(i, j) = 1$ if and only if $X_M[j] \geq n - i + 1$.
    
    Thus, any algorithm $\cA$ for identifying a list in $\cT_n$ can be turned into an algorithm, with the exact same number of queries, for identifying a matrix in $\cC_n$ by replacing every query ``Is $X[j] \geq i$?'' to the hidden list $X$ with the query ``Is $M(n-i+1,j) = 1$?'' to the hidden matrix $M$.

    The argument to show that an algorithm for identifying a matrix in $\cC_n$ can be turned into an algorithm for identifying a list in $\cT_n$ is similar, and we omit it.
\end{proof}

\subsubsection{Perfectly interleaved bipartite sorting}\label{sec:perfectly-interleaved}
Let $R'$ and $B'$ be two $n$-element lists consisting of pairwise distinct numbers such that when the elements from both lists are sorted there are no two consecutive elements from the same list. Without loss of generality, one may assume that $R'$ consists of all even numbers in $[2n]$ and $B'$ consists of all odd numbers in $[2n]$; the elements in the two lists are ordered arbitrarily. 
The \emph{perfectly interleaved bipartite sorting} problem is to sort the elements in the two lists into one total order by making comparisons only between elements from different lists; in other words, an algorithm that solves this problem can only use an oracle that answers queries 
``Is $R'[i] > B'[j]$?'' for any $i,j \in [n]$. 

From now on, we will work with an equivalent representation of $R'$, $B'$ by two lists $R$ and $B$ each consisting of the numbers in $[n]$, where for every $i \in [n]$ we have $R[i]=s$ if and only if $R'[i] = 2s$, and $B[i]=s$ if and only if $B'[i] = 2s-1$. In particular, we have $R[i] \geq B[j]$ if and only if $R'[i] > B'[j]$. This is precisely the classical nuts-and-bolts setting, where an algorithm does not have access to equality checks.

The cost of the algorithm is the number of queries made to the oracle in the worst case.
Let $\cB_n$ denote the family of all pairs of lists $(R,B)$ as above.
Using the query complexity framework from Section \ref{subsec:prelims-query-complexity},
we define the deterministic, randomized, and quantum query complexities of bipartite sorting on instances from $\cB_n$.
We denote them $\sD^{>}(\cB_n)$, $\sR^{>}(\cB_n)$, and $\sQ^{>}(\cB_n)$, respectively.

We claim that the query complexity of identifying a hidden half graph from $\cH_n$ is the same as that of bipartite sorting of instances from $\cB_n$.

\begin{lemma}\label{lem:bip-sorting-thr-equiv-permuted}
    Let $n$ be a positive integer. Then for all $\sC \in \cbra{\sD, \sR, \sQ}$,
    $$
    \sC^{>}(\cB_n) = \sC(\cH_n).
    $$
\end{lemma}
\begin{proof}
    The proof strategy is similar to that of Lemma \ref{lem:sorting-thr-equiv-colpermuted}, i.e., we establish a bijection between $\cB_n$ and $\cH_n$, and show that the queries are equivalent.

    For the bijection, for every instance $(R,B) \in \cB_n$ we define a pair of permutations $\pi_R$ and $\pi_B$ of $[n]$ such that $\pi_R(i) = R[i]$ and $\pi_B(i) = B[i]$ for every $i \in [n]$. Our bijection associates to $(R,B)$ the matrix $M_{(R,B)} \in \cH_n$ that is obtained from $L_n$ by permuting its rows according to the permutation $\pi_R^{-1}$ and permuting its columns according to the permutation $\pi_B^{-1}$. For a matrix $M \in \cH_n$, we denote by $(R_M,B_M)$ the instance in $\cB_n$ corresponding to $M$ under this bijection. 
    
    By definition, we have $M_{(R,B)}(i,j) = L_n(\pi_R(i),\pi_B(j)) = L_n(R[i],B[j])$ for all $i,j \in [n]$. Since $L_n(i,j) = 1$ if and only if $i \geq j$, we have that $M_{(R,B)}(i,j) = 1$ if and only if $R[i] \geq B[j]$. Thus, the query ``Is $R[i] \geq B[j]$?'' on instance $(R,B) \in \cB_n$ is equivalent to the query ``Is $M_{(R,B)}(i,j) = 1$?'' on $M_{(R,B)} \in \cH_n$.
    Therefore, any algorithm $\cA$ for solving the perfectly interleaved bipartite sorting problem on instances in $\cB_n$ can be turned into an algorithm, with the exact same number of queries, for identifying a matrix in $\cH_n$ by replacing every query ``Is $R[i] \geq B[j]$?'' to the lists $(R,B) \in \cB_n$ with the query ``Is $M(i,j) = 1$?'' to the hidden matrix $M \in \cH_n$. 

    The argument to show that an algorithm for identifying a matrix in $\cH_n$ can be turned into an algorithm for solving the perfectly interleaved bipartite sorting problem on instances in $\cB_n$ is similar, and we omit it.
\end{proof}

\subsection{Learning column-permuted half graphs}\label{subsec:column-permuted-half-graph-learning}

In this section, we prove Theorem~\ref{thm-col-permuted-half-graphs}, restated below.
\mainResultColPermutedHalfGraphs*

We require the adversary method, due to Ambainis~\cite{Amb02}, to show our quantum query lower bounds here.
\begin{lemma}[\cite{Amb02}]\label{lem:adversary}
    Let $\cD \subseteq \zone^N$, let $\cS$ be a finite set, and $f : \cD \to \cS$ be a function. Let $X, Y \subseteq \cD$ be two sets of inputs such that $f(x) \neq f(y)$ if $x \in X$ and $y \in Y$. Let $\cR \subseteq X \times Y$ be nonempty, and satisfy:
    \begin{itemize}
        \item For every $x \in X$ there exist at least $m$ different $y \in Y$ such that $(x, y) \in \cR$.
        \item For every $y \in Y$ there exist at least $m'$ different $x \in X$ such that $(x, y) \in \cR$.
        \item For every $x \in X$ and $i \in [N]$, there are at most $\ell$ different $y \in Y$ with $(x, y) \in \cR$ and $x_i \neq y_i$.
        \item For every $y \in Y$ and $i \in [N]$, there are at most $\ell'$ different $x \in X$ with $(x, y) \in \cR$ and $x_i \neq y_i$.
    \end{itemize}
    Then any quantum query algorithm that computes $f$ with success probability at least $2/3$ uses $\Omega\rbra{\sqrt\frac{m m'}{\ell\ell'}}$ quantum queries.
\end{lemma}
\begin{proof}[Proof of Theorem~\ref{thm-col-permuted-half-graphs}]
The upper bounds follow easily from the fact that each column can be learned with $O(\log n)$ queries using binary search to find the top-most $1$-entry of that column.

For the deterministic and randomized lower bounds, first recall that $\sC(\cC_n) = \sC(\mathsf{Id}_{\cC_n})$ for all $\sC \in \cbra{\sD, \sR, \sQ}$, where $\mathsf{Id}_{\cC_n} : \cC_n \to \cC_n$ is the identity function.
Second, note that each query in an algorithm has at most two possible outcomes (an element of $\zone$), and that $|\cC_n| = n!$. Thus, Lemma~\ref{lem:large-range-query-lower-bound} implies that $\sC(\cC_n) = \Omega(\log n!) = \Omega(n \log n)$ for $\sC \in \cbra{\sD, \sR}$.

For the quantum lower bound, we use Lemma~\ref{lem:adversary}. Define
\begin{align*}
    \cR = \{(M, M') : M' \textnormal{ can be obtained from } M \textnormal{ by interchanging the column}\\
    \textnormal{of Hamming weight } j \textnormal{ with that of Hamming weight } j + 1 \textnormal{ for some } j \in [n - 1]\}.
\end{align*}
For a fixed $M$, there are $n-1$ matrices $M'$ with $(M,M')\in\cR$, corresponding to the $n-1$ choices of $j$. Thus, $m=n-1$. Similarly, $m'=n-1$.

For a fixed $M$ and $(i,j)$, there is at most one $M'$ with $(M,M') \in \cR$ and $M(i,j) \neq M'(i,j)$: if $M(i,j)=0$, the only candidate $M'$ must swap the $j$th column of $M$ with the column with Hamming weight one higher than that of the $j$th column, and similarly if $M(i,j)=1$. Thus, $\ell=1$. Similarly, $\ell'=1$. Thus, Lemma~\ref{lem:adversary} implies $\sQ(\cC_n) = \Omega(n)$.
\end{proof}

\subsection{Learning general half graphs}\label{subsec:row-column-permuted-half-graph-learning}

In this section, we prove Theorem \ref{thm-general-half-graphs}.
For our quantum upper bound, we require the subroutine given by the following lemma.
Below, for bit-strings $x, y$, let $x \oplus y$ denote the bitwise XOR of $x$ and $y$.
\begin{lemma}\label{lem: bbht useful}
    Let $n$ be a positive integer and let $x \neq y \in \zone^n$ be such that either $x \leq y$ or $x \geq y$ entrywise. Then there exists a quantum query algorithm (with query access to the entries of $x$ and $y$) that uses $O\left(\sqrt{\frac{n}{|x \oplus y|}}\right)$ queries in expectation, and outputs the following:
    \begin{itemize}
        \item $x \leq y$ and an index $i \in [n]$ with $x_i = 0$ and $y_i = 1$, or
        \item $x \geq y$ and an index $i \in [n]$ with $x_i = 1$ and $y_i = 0$.
    \end{itemize}
\end{lemma}
We remark that the algorithm need not know $x \oplus y$ in advance.
This lemma follows as a relatively straightforward application of Lemma~\ref{lem:bbht}: use the algorithm from Lemma~\ref{lem:bbht} on $x \oplus y$ (the bitwise XOR of $x$ and $y$). The expected number of queries made by this algorithm is $O\left(\sqrt{\frac{n}{|x \oplus y|}}\right)$ since querying the $i$'th bit of $x \oplus y$ can be done using 2 queries. The algorithm outputs an index $i \in [n]$ with $x_i \neq y_i$. Since we are promised that either $x \leq y$ or $x \geq y$, querying the value of $x_i$ (or $y_i$) tells us which of these is the case.

Using a similar argument with Lemma~\ref{lem:sample} instead of Lemma~\ref{lem:bbht}, we obtain the following classical analog of Lemma~\ref{lem: bbht useful}.

\begin{lemma}\label{lem: sample useful}
    Let $n$ be a positive integer and let $x \neq y \in \zone^n$ be such that either $x \leq y$ or $x \geq y$ entrywise. Then there exists a randomized query algorithm (with query access to the entries of $x$ and $y$) that uses $O\left({\frac{n}{|x \oplus y|}}\right)$ queries in expectation, and outputs the following:
    \begin{itemize}
        \item $x \leq y$ and an index $i \in [n]$ with $x_i = 0$ and $y_i = 1$, or
        \item $x \geq y$ and an index $i \in [n]$ with $x_i = 1$ and $y_i = 0$.
    \end{itemize}
\end{lemma}

Before proceeding to the proof of Theorem~\ref{thm-general-half-graphs}, we remark that  
an argument similar to that in the proof of Lemma~\ref{lem:matchinglb-det} cannot be used to show a deterministic lower bound for learning graphs from $\cH_n$. Specifically, we show that it is possible for the all-0 path (a similar argument holds for the all-1 path) of a deterministic algorithm to have length $O(n)$. We show this in Appendix \ref{appedix-rule-out-lower-bound-approach}.

\mainResultGeneralHalfGraphs*

\begin{proof}    
    Our randomized and quantum algorithms are modifications of quicksort; they have the same structure and differ only in one step, so we describe these algorithms together. In short, we find the row permutation and column permutation separately. Our algorithm to find the row permutation (the algorithm to find the column permutation is essentially the same and we omit this\footnote{Once the algorithm finds the row permutation, there is a much easier algorithm to find the column permutation: use a $\log n$-cost binary search for each column to determine its position. Since this does not affect the complexity of the overall algorithm, we do not discuss this in more detail here.}) has the same structure as the randomized QuickSort algorithm, where the elements in the list are the entire rows (so querying an element takes cost $n$). Just as in the textbook randomized quicksort algorithm, we randomly choose a pivot row, compare all other rows with the pivot to figure out which rows are smaller than and which rows are larger than the pivot. The comparisons with the pivot are done using the algorithm from Lemma~\ref{lem: bbht useful} in the quantum algorithm and using the algorithm from Lemma~\ref{lem: sample useful} in the randomized algorithm.\footnote{Since all pairwise rows are unequal throughout the algorithm, Lemmas~\ref{lem: bbht useful} and~\ref{lem: sample useful} can be applied.} We then recurse into each subproblem. The analysis of our algorithm is a modification of a textbook analysis of randomized QuickSort.
        \begin{algorithm}[h]
        \begin{algorithmic}[1]
            \State \textbf{Input:} List of pairwise unequal rows $R = [R_1, \dots, R_n]$ with each $R_i \in \zone^n$, and the promise that for all $i, j \in [n]$, either $R_i \leq R_j$ or $R_j \leq R_i$ entrywise.
            \If{$n = 1$} 
                \State Return $R$
            \Else
                \State $i \gets$ uniformly random element of $[n]$\label{line:random}
                \Comment{Pick a random row to be the `pivot'}
                \State $L \gets \emptyset$, $G \gets \emptyset$
                \Comment{$L$ and $G$ will store rows smaller and larger than $R_i$, respectively}
                \State Query $R_i$ \label{line:queryrow}
                \Comment{This uses $n$ queries}
                \For{$j \in [n] \setminus \cbra{i}$}\label{line: for loop}
                    \State Compare $R_j$ with $R_i$\label{line:compare}
                    \Comment{Lemma~\ref{lem: bbht useful} (quantum), or Lemma~\ref{lem: sample useful} (randomized)}\label{line: comparison}
                    \If{$R_j < R_i$}
                        \State Remove 0-indices of $R_i$ \label{line:remove0}
                        \Comment{All 0's in $R_i$ must be 0's in $R_j$}
                        \State $L \gets L \cup \cbra{R_j}$
                    \Else
                    \Comment{$R_j > R_i$}
                        \State Remove 1-indices of $R_i$ \label{line:remove1}
                        \Comment{All 1's in $R_i$ must be 1's in $R_j$}
                        \State $G \gets G \cup \cbra{R_j}$
                    \EndIf
                \EndFor
                \State Return $\cA(L), R_i, \cA(G)$\label{line:recurse}
                \Comment{We know the position of $R_i$ now. Recursively solve $L, G$}
            \EndIf
        \caption{Query algorithm $\cA$ for finding row permutation}
        \label{algo:qqrowpermutation}
        \end{algorithmic}
        \end{algorithm}
        Lines~\ref{line:remove0} and~\ref{line:remove1} ensure that the problems that we recurse into correspond to square matrices.
        Figure \ref{fig:pivot} depicts a single iteration of the \textbf{for} loop in Line~\ref{line: for loop} in Algorithm~\ref{algo:qqrowpermutation}. As we know the entire pivot row from Line~\ref{line:queryrow}, by permuting columns we may assume that the pivot row has the form $1^k0^{n-k}$. Observe that all rows smaller than the pivot must have 0's in all the positions where the pivot has 0's. Similarly, all rows larger than the pivot must have 1's in all the positions where the pivot has 1's. Thus, the problems $L$ and $G$ are smaller instances of the original problem. We now analyze the expected query complexity of Algorithm~\ref{algo:qqrowpermutation}.
            
    \begin{figure}[h!]
        \centering
        \begin{tikzpicture}
        
        \matrix (m) [matrix of nodes,left delimiter={[},right delimiter={]},ampersand replacement=\&,
                     nodes in empty cells, nodes={minimum width=1.5em, minimum height=1.5em, anchor=center}] {
          \& \& \& \& \& \& \& \\
          \& \& \& \& \& \& \& \\
          \& \& \& \& \& \& \& \\
          \& \& \& \& \& \& \& \\
          \& \& \& \& \& \& \& \\
          \& \& \& \& \& \& \& \\
          \& \& \& \& \& \& \& \\
        };
        
        \node at (m-4-1) {1};
        \node at (m-4-2) {1};
        \node at (m-4-3) {...};
        \node at (m-4-4) {1};
        \node at (m-4-5) {0};
        \node at (m-4-6) {0};
        \node at (m-4-7) {...};
        \node at (m-4-8) {0};
        
        \node at (m-2-2) {\huge $L$};
        \node at (m-6-7) {\huge $G$};
        \node at (m-2-7) {\huge $0$};
        \node at (m-6-2) {\huge $1$};
        
        \draw[solid, gray!50] (m-3-1.south west) -- (m-3-8.south east); 
        \draw[solid, gray!50] (m-4-1.south west) -- (m-4-8.south east); 
        \draw[solid, gray!50] (m-1-4.north east) -- (m-7-4.south east); 

        \node[left=0.5cm of m-4-1] {Pivot $\rightarrow$};
        
        \end{tikzpicture}
        \caption{Illustration of a single iteration of the \textbf{for} loop in Line~\ref{line: for loop} in Algorithm~\ref{algo:qqrowpermutation}.}
        \label{fig:pivot}
    \end{figure}

    \paragraph*{Quantum algorithm:} Let $T(n)$ denote the expected query complexity of finding the hidden row permutation of a row-permutation of $L_n$ (where the columns may be arbitrarily permuted as well). By Algorithm~\ref{algo:qqrowpermutation}, each of the pivots is chosen with equal likelihood. We have
    \begin{equation}\label{eqn: recurrence}
        T(n) = \frac{1}{n}\sum_{i = 1}^n\left[n + \sum_{j = 1}^{i - 1}O\left(\sqrt{\frac{n}{i - j}}\right) + \sum_{j = i + 1}^nO\left(\sqrt\frac{n}{j - i}\right) + T(i - 1) + T(n - i)\right].
    \end{equation}
    Here, the outermost $1/n$ corresponds to the probability of each row being chosen as the pivot in Line~\ref{line:random}. The outermost summand being set to $i$ corresponds to the event of the $i$'th row (in the original matrix $L_n$ before applying the hidden permutation) being chosen as the pivot. The first term of $n$ is the cost of querying the entire row in Line~\ref{line:queryrow}. Next, recall that each row is compared with the pivot in the algorithm exactly once (Line~\ref{line:compare}). Lemma~\ref{lem: bbht useful} shows that the comparison between $R_i$ and $R_j$ can be done with an expected quantum query complexity of $\sqrt{\frac{n}{|R_i \oplus R_j|}}$. The next two terms in Equation~\eqref{eqn: recurrence} arise due to these comparisons since there is a unique row of every possible Hamming weight from 1 to $n$ in the matrix. The final  two terms, $T(i - 1)$ and $T(n - i)$, are the costs of the algorithm run recursively on $L$ and $G$, respectively (Line~\ref{line:recurse}).

    Replacing the second and third terms in the brackets of Equation~\eqref{eqn: recurrence} by a naive upper bound of $O\left(\sum_{j = 1}^n \sqrt{\frac{n}{j}}\right)$ each, we obtain    
    \begin{align*}
        T(n) & = \frac{1}{n}\sum_{i = 1}^n\left[n + O\left(\sum_{j = 1}^{n}\sqrt{\frac{n}{j}}\right) + T(i - 1) + T(n - i)\right]\\
        & = \frac{1}{n}\sum_{i = 1}^n\left[n + O(n) + T(i - 1) + T(n - i)\right] = \frac{1}{n}\sum_{i = 1}^n\left[O(n) + T(i - 1) + T(n - i) \right].
    \end{align*}
    Above, the second equality follows since $\sum_{j = 1}^n \frac{1}{\sqrt{j}} = \Theta(\sqrt{n})$ using standard techniques. This is now precisely the recurrence that analyzes the expected running time of the randomized QuickSort algorithm and solves to $T(n) = O(n \log n)$. We omit this proof and refer the reader to standard texts or lecture notes for details, see, for example,~\cite{CLRS}.
    
    \paragraph*{Randomized algorithm:} The analysis here follows along the same lines as that for the quantum algorithm so we skip some of the steps that are exactly the same. The only difference between the algorithms is that comparisons between two rows (in Line~\ref{line:compare}) is done slower in the randomized one, using Lemma~\ref{lem: sample useful} in place of Lemma~\ref{lem: bbht useful}. Letting $T(n)$ denote the expected query complexity, we have the analog of Equation~\eqref{eqn: recurrence} as
    \begin{align*}
        T(n) = \frac{1}{n}\sum_{i = 1}^n\left[n + \sum_{j = 1}^{i - 1}O\left(\frac{n}{i - j}\right) + \sum_{j = i + 1}^nO\left(\frac{n}{j - i}\right) + T(i - 1) + T(n - i)\right].
    \end{align*}
    Replacing the second and third terms by a naive upper bound of $O\left(\sum_{j = 1}^n \frac{n}{j}\right)$ each, we obtain
        \begin{align*}
        T(n) & = \frac{1}{n}\sum_{i = 1}^n\left[n + O\left(\sum_{j = 1}^{n}\frac{n}{j}\right) + T(i - 1) + T(n - i)\right]\\
        & = \frac{1}{n}\sum_{i = 1}^n\left[n + O(n \ln n) + T(i - 1) + T(n - i)\right] = \frac{1}{n}\sum_{i = 1}^n\left[O(n \ln n) + T(i - 1) + T(n - i) \right]\\
        & = O(n \ln n) + \frac1n \sum_{i = 1}^n [T(i - 1) + T(n - i)] \leq O(n \ln n) + \frac{2}{n} \sum_{i = 1}^{n - 1} T(i).
    \end{align*}
    Here we used the fact that $\sum_{j = 1}^n \frac1j = O(\ln n)$, where $\ln$ denotes the natural logarithm. Let $k > 0$ be a constant such that
    \begin{equation}\label{eq:T_n-eneq}
        T(n) \leq kn \ln n + \frac{2}{n} \sum_{i = 1}^{n - 1} T(i),                    
    \end{equation}
    We prove by induction on $n$ that $T(n) \leq 2k n\ln^2 n$ holds for all natural $n$.
    Indeed, assuming this inequality holds for all $i < n$, from (\ref{eq:T_n-eneq}) we have 
    \begin{align*}
        T(n) & \leq kn \ln n + \frac{2}{n} \sum_{i = 1}^{n - 1} T(i) \leq k \left( n \ln n + \frac4n \sum_{i = 1}^{n - 1} i \ln^2 i \right) \\
        & \leq k \left( n \ln n + \frac4n \int_{1}^n x \ln^2 x dx \right) \leq k \left( n\ln n + \frac4n \left[\frac{n^2}{4} + \frac{2n^2 \ln^2 n}{4} - \frac{2n^2 \ln n}{4} - \frac14\right] \right)\\
        & \leq k\left( n\ln n + n + 2n \ln^2 n - 2n \ln n - \frac{1}{n} \right) \leq 2kn \ln^2 n.
    \end{align*}
    This proves $T(n) = O(n \log^2 n)$.
\end{proof}

We conclude this section with a tight lower bound on quantum query complexity of learning half graphs.
\begin{lemma}\label{lem:quantum-lower-half-graphs}
    $\sQ(\cH_n) = \Omega(n \log n)$.
\end{lemma}

\begin{proof}
    We require the quantum comparison-based sorting lower bound of $\Omega(n \log n)$ due to H\o yer, Neerbek and Shi~\cite{HNS02}. From the argument in~\cite[Section~5]{HNS02}, one may note that their lower bound applies to the problem of learning a hidden matrix that can be obtained from a lower triangular matrix by applying a permutation to the rows and the same permutation to the columns. Since our class $\cH_n$ includes all these matrices, a lower bound of $\Omega(n \log n)$ follows from~\cite[Theorem~2]{HNS02}.
\end{proof}

\section{Conclusions}\label{sec:conc}
Our results leave some natural open questions. Perhaps the most immediate is to give a self-contained deterministic algorithm using $O(n\log n)$ queries for learning general half graphs. As we mentioned in the discussion following Theorem~\ref{thm-general-half-graphs}, we were unable to verify the claimed $O(n\log n)$ bound of Koml{\'{o}}s, Ma, and Szemer{\'{e}}di~\cite{KMS98}. Another natural question is to close the gap between the $\Omega(n)$ lower bound and the $O(n\log n)$ upper bound for the quantum query complexity of learning column-permuted half graphs. We conjecture that the latter bound is tight and suspect that the $\Omega(n \log n)$ lower-bound argument for comparison-based sorting~\cite[Theorem~2]{HNS02} can be generalized to our setting.

More broadly, our arguments make substantial use of the specific structure of matchings and half graphs. The lower bounds for matchings exploit the natural operation of swapping columns of a permutation matrix, while our algorithms for half graphs exploit the ordering structure of their neighborhoods. It would be interesting to understand whether similarly simple ideas can be useful for other natural graph classes in the edge-query model.

\section*{Acknowledgements}
We thank the anonymous DMTCS reviewers for their helpful comments and suggestions.

\bibliography{bibo}

\appendix

\section{Query lower bound for functions with large range}

In this section, we prove Lemma~\ref{lem:large-range-query-lower-bound}, restated below.
\countingLowerBound*
\begin{proof}  
    We prove the randomized lower bound, which implies the deterministic lower bound. We use Yao's principle~\cite{Yao77}. For each $r \in \cR$, let $d_r \in \cD$ be an arbitrary representative such that $f(d_r) = r$. Define $\cD' = \cbra{d_r : r \in \cR}$. Consider the distribution $\mu$ to be the uniform distribution over all inputs of $\cD'$. For a suitable constant $c > 0$ to be fixed later, it suffices to show that for every deterministic algorithm $\cA$ of query complexity less than $c\log|\cR|$, $\cA$ must make an error on at least a 1/3 fraction of all inputs drawn from $\mu$. Each internal node of the decision tree corresponding to $\cA$ has at most two children: one corresponding to each possible answer to the query in $\zone$. Towards a contradiction, fix an algorithm $\cA$ of query complexity less than $c\log|\cR|$ and that errs on less than a $1/3$ fraction of all inputs. Since each internal node has only two children, the number of leaves of the decision tree representing $\cA$ is at most $2^{c\log|\cR|} < \frac{2}{3}|\cR|$. This means $\cA$ must err on at least $\frac13 |\cR|$ inputs in $\cD'$, and hence make an error of at least $\frac13$ with respect to $\mu$, which is a contradiction. This proves the randomized lower bound.
\end{proof}

\section{Ruling out a lower bound approach for perfectly interleaved bipartite sorting}\label{appedix-rule-out-lower-bound-approach}
In this section we show that the lower bound technique in Lemma~\ref{lem:matchinglb-det} cannot be used to give an $\Omega(n^2)$ lower bound for Theorem~\ref{thm-general-half-graphs}. Specifically, we show that there is a certificate of size $O(n)$ that consists only of 0s. A similar proof also shows existence of a certificate of size $O(n)$ that consists only of 1s, but we omit this here.
\begin{lemma}\label{lem: inductive proof}
    Let $n$ be a positive integer and let $M$ be obtained by applying a row permutation and a column permutation to $L_n$. Then there exists a subset $S$ of the 0-entries of $M$ such that $|S| = O(n)$ and $M$ is the unique matrix that is a row and column permutation of $L_n$ that is consistent with all entries of $S$ being 0.
\end{lemma}
We remark that the same statement also holds true if $S$ is a subset of the 1-entries. The proof is extremely similar to this one, and we omit it.
\begin{proof}
    Without loss of generality (renaming rows and columns suitably), we may assume that $M = L_n$. Define $S = \cbra{(i, j) : i, j \in [n] \textnormal{ and } j \in \cbra{i+1, i+2}}$. Clearly, all of these entries are 0 in $M$, and $|S| = O(n)$. Equation~\eqref{eq:example} shows the entries of $S$ when $n = 6$.
    \begin{equation}\label{eq:example}
\begin{blockarray}{ccccccc}
~ & 1 & 2 & 3 & 4 & 5 & 6\\
\begin{block}{c(cccccc)}
  1 & ~ & 0 & 0 & ~ & ~ & ~\\
  2 & ~ & ~ & 0 & 0 & ~ & ~\\
  3 & ~ & ~ & ~ & 0 & 0 & ~\\
  4 & ~ & ~ & ~ & ~ & 0 & 0\\
  5 & ~ & ~ & ~ & ~ & ~ & 0\\
  6 & ~ & ~ & ~ & ~ & ~ & ~\\
\end{block}
\end{blockarray}
\end{equation}
In the remainder of the proof, we show that the identity permutations on the rows and columns are the unique row and column permutations of $L_n$ that are consistent with the 0s in $S$. We prove this by giving an inductive process that proceeds one row and one column at a time from top to bottom and left to right. During this process we ensure that each row and column must be filled uniquely, and this must be consistent with $L_n$. The first step is as follows: 
\begin{itemize}
    \item Observe that every column other than the first one has at least one 0. Thus, the first column must be the all-1 column.
\end{itemize}

For each $i \in \cbra{2, \dots, n}$, Step $i$ is the following, assuming that we've ascertained so far that the first $i-1$ columns and the first $i - 2$ rows are consistent with $L_n$. We treat the first $i-1$ columns and first $i-2$ rows to be fixed.
\begin{itemize}
    \item After the preceding steps, the zero-counts of the remaining columns must be $i-1,i,\dots,n-1$. Column $i$ already has its first $i-1$ entries fixed to 0, whereas every other unfixed column has at least $i$ entries fixed to 0. Hence Column $i$ is the unique remaining column with exactly $i-1$ zeros, and all its other entries must be 1.
    \item The weights of the remaining rows must now be $i-1,i,\dots,n$. Row $i-1$ already has its first $i-1$ entries fixed to 1, whereas every other unfixed row has at least $i$ entries fixed to 1. Hence Row $i-1$ is the unique remaining row of weight $i-1$, and all its other entries must be 0.
\end{itemize} 
Putting things together, this means that $L_n$ is the unique matrix that is consistent with the 0s in $S$. For convenience, we include a pictorial description of one step of the inductive process with $n = 6$ and $i = 3$ in Figure~\ref{fig:example}.
\begin{figure}[hbt]
    \centering
\begin{blockarray}{ccccccc}
~ & 1 & 2 & 3 & 4 & 5 & 6\\
\begin{block}{c(cccccc)}
  1 & 1 & 0 & 0 & 0 & 0 & 0\\
  2 & 1 & 1 & 0 & 0 & \textcolor{blue}{0} & \textcolor{blue}{0}\\
  3 & 1 & 1 & \textcolor{red}{1} & 0 & 0 & ~\\
  4 & 1 & 1 & \textcolor{red}{1} & ~ & 0 & 0\\
  5 & 1 & 1 & \textcolor{red}{1} & ~ & ~ & 0\\
  6 & 1 & 1 & \textcolor{red}{1} & ~ & ~ & ~\\
\end{block}
\end{blockarray}
    \caption{Depiction of inductive process in the Proof of Lemma~\ref{lem: inductive proof} with $n = 6$ and $i = 3$. The entries colored in black are those we've deduced after the first two steps, fixing Column 1, Column 2, and Row 1. Entries colored in \textcolor{red}{red} are those fixed first in this step, and entries colored in \textcolor{blue}{blue} are those fixed second in this step. The reasoning for fixing the \textcolor{red}{red} entries is as follows: Column 3 is now the only un-fixed column with at most two 0s. This means Column 3 is in its correct place, and the remaining entries of Column 3 must be 1. The reasoning for fixing the \textcolor{blue}{blue} entries is: Row 2 is now the only row with two 1s, implying that Row 2 is in the correct place.}
    \label{fig:example}
\end{figure}
\end{proof}
\end{document}